\documentclass[runningheads]{llncs}
\title{Solving Problems on Generalized Convex Graphs via Mim-Width\thanks{Brettell and Paulusma received support from the Leverhulme Trust (RPG-2016-258).
Bonomo received support from CONICET (PIP 11220200100084CO) and UBACyT (20020170100495BA and 20020160100095BA).
Brettell was also supported by a Rutherford Foundation Postdoctoral Fellowship, administered by the Royal Society Te Ap\={a}rangi. An extended abstract of the paper appeared in the proceedings of WADS 2021~\cite{BBMP21}.}}
\titlerunning{Solving Problems on Generalized Convex Graphs via Mim-width}

\author{Flavia Bonomo-Braberman\inst{1}\orcidID{0000-0002-9872-7528}
\and\\
Nick Brettell\inst{2}\orcidID{0000-0002-1136-418X}
\and\\
Andrea Munaro\inst{3}\orcidID{0000-0003-1509-8832}
\and\\
Dani\"el Paulusma\inst{4}\orcidID{0000-0001-5945-9287}}

\authorrunning{F. Bonomo-Braberman, N. Brettell, A. Munaro, D. Paulusma}

\institute{Universidad de Buenos Aires. Facultad de Ciencias
Exactas y Naturales. Departamento de Computaci\'on. Buenos Aires,
Argentina. / CONICET-Universidad de Buenos Aires. Instituto de
Investigaci\'on en Ciencias de la Computaci\'on (ICC). Buenos
Aires, Argentina, \email{fbonomo@dc.uba.ar} \and School of
Mathematics and Statistics,
Victoria University of Wellington, New Zealand,  \email{nick.brettell@vuw.ac.nz}
\and
School of Mathematics and Physics,
Queen's University Belfast, UK, \email{a.munaro@qub.ac.uk}
\and
Department of Computer Science,
Durham University, UK, \email{daniel.paulusma@durham.ac.uk}}

\usepackage[T1]{fontenc}
\usepackage{amsmath,amssymb,xspace,xcolor,tikz,graphicx,tabularx,tablefootnote,boxedminipage,comment}
\usetikzlibrary{fit,automata,arrows}
\usetikzlibrary{patterns,hobby}
\usepackage{pgfplots}
\usepackage{subcaption}
\usepackage{hyperref}
\usepackage{cleveref}

\newcommand\faketheorem[1]{{\bfseries#1}}
\DeclareMathOperator{\thin}{thin}
\DeclareMathOperator{\pthin}{pthin}

\newcommand{\NP}{{\sf NP}}

\newcommand{\cutmim}{\mathrm{cutmim}}
\newcommand{\cutsim}{\mathrm{cutsim}}
\newcommand{\mimw}{\mathrm{mimw}}
\newcommand{\simw}{\mathrm{simw}}

\DeclareMathOperator{\pw}{pw}

\usetikzlibrary{shapes,shapes.geometric,arrows,fit,calc,positioning,automata,arrows}
\usetikzlibrary{patterns,hobby}
\usepackage{pgfplots}
\pgfplotsset{compat=1.15}

\tikzstyle{line}=[draw]

\tikzset{ n1/.style={circle,scale=1.0},
n2/.style={circle,fill=black,scale=0.5},
n3/.style={circle,draw,fill=black,draw=black,text=white,scale=0.9},
e1/.style={line width=0.1mm}, e3/.style={draw=black,line
width=0.7mm}, inter/.style={line width=0.85mm},
e2/.style={draw=black,line width=0.5mm}, c1/.style={line
width=0.3mm}, c2/.style={line width=0.2mm} }

\newcommand{\vertex}[4][black]{
    \draw[#1, fill=#1, inner sep=0pt] (#2, #3) circle (0.12) node(#4){};
}

\begin{document}

\maketitle

\begin{abstract}
A bipartite graph $G=(A,B,E)$ is ${\cal H}$-convex, for some family of graphs ${\cal H}$, if there exists a graph $H\in {\cal H}$ with $V(H)=A$ such that the set of neighbours in $A$ of each $b\in B$ induces a connected subgraph of~$H$.
Many $\mathsf{NP}$-complete problems, including problems such as $\textsc{Dominating Set}$, $\textsc{Feedback Vertex Set}$, $\textsc{Induced Matching}$ and $\textsc{List $k$-Colouring}$,
become polynomial-time solvable for ${\mathcal H}$-convex graphs when ${\mathcal H}$ is the set of paths. In this case, the class of ${\mathcal H}$-convex graphs is known as the class of convex graphs. The underlying reason is that the class of convex graphs
 has bounded mim-width. We extend the latter result to families of ${\mathcal H}$-convex graphs where (i) ${\mathcal H}$ is the set of cycles, or (ii) ${\mathcal H}$ is the set of trees with bounded maximum degree and a bounded number of vertices of degree at least~$3$.
As a consequence, we can strengthen a large number of results on generalized convex graphs known in the literature via one general and relatively short proof.
To complement result (ii), we show that the mim-width of ${\mathcal H}$-convex graphs is unbounded if ${\mathcal H}$ is the set of trees with arbitrarily large maximum degree or an arbitrarily large number of vertices of degree at least~$3$.
In this way we are able to determine complexity dichotomies for the aforementioned graph problems.
We prove our results via a more refined width-parameter analysis. This yields an even clearer picture of
 which width parameters are bounded for classes of ${\cal H}$-convex graphs.
\end{abstract}

\section{Introduction}\label{sec:intro}

Many computationally hard graph problems can be solved efficiently if we place constraints on the input.
Instead of solving individual problems in an ad hoc way we may try to decompose the vertex set of the input graph into large sets of ``similarly behaving'' vertices and to exploit this decomposition for an algorithmic speed up that works for many problems simultaneously. This requires some notion of an ``optimal'' vertex decomposition, which depends on the type of vertex decomposition used and which may relate to the minimum number of sets or the maximum size of a set in a vertex decomposition. An optimal vertex decomposition gives us the ``width'' of the graph.

A graph class has {\it bounded width} if every graph in the class has width at most some constant $c$. Boundedness of width is often the underlying reason why a graph-class-specific algorithm runs efficiently: in this case, the proof that the algorithm is efficient for some special graph class reduces to a proof showing that the width of the class is bounded by some constant. We will give examples, but also refer to the
surveys~\cite{DJP19,HOSG08,Ja20,KLM09,Va12}
 for further details and examples.

Width parameters differ in strength. A width parameter $p$ {\it dominates} a width parameter~$q$ if there is a function~$f$ such that $p(G)$ is at most $f(q(G))$ for every graph~$G$. If $p$ dominates $q$ but $q$ does not dominate $p$, then we say that $p$ is {\it more powerful} than $q$. If both $p$ and $q$ dominate each other, then $p$ and $q$ are {\it equivalent}.
If neither $p$ is more powerful than $q$ nor $q$ is more powerful than $p$, then $p$ and $q$ are {\it incomparable}.
If $p$ is more powerful than $q$, then
 the class of graphs for which $p$ is bounded is larger than the class of graphs for which $q$ is bounded and so efficient algorithms for bounded $p$ have greater applicability with respect to the graphs under consideration.
 The trade-off is that  fewer problems exhibit an efficient algorithm for the parameter $p$, compared to the parameter $q$.

This notion of powerfulness leads to a large hierarchy of width parameters,  in which new width parameters continue to be defined,
for example, graph functionality~\cite{AAL19} in 2019 and twin-width~\cite{BKTW20} in 2020.
The well-known parameters boolean-width, clique-width, module-width and rank-width are equivalent to each other~\cite{BTV11,OS06,Ra08}. They are more powerful than the equivalent parameters branch-width and treewidth~\cite{CO00,RS91,Va12} but less powerful than mim-width~\cite{Va12}, which is less powerful than sim-width~\cite{KKST17}. To give another example, thinness is more powerful than path-width~\cite{M-O-R-C-thinness}, but less powerful than mim-width and incomparable to clique-width or treewidth~\cite{BE19}.

For each group of equivalent width parameters, a growing set of \NP-complete problems is known to be tractable on graph classes of bounded width. Proving the latter for some graph class often immediately tells us that many problems are tractable for that class without the need for constructing algorithms for each problem.
However, there are still large families of graph classes, and many width parameters, for which it is not known whether the class has bounded width.

\medskip
\noindent
{\bf Our Focus.}
We consider the relatively new width parameter {\it mim-width}, which we define below.
Recently, we showed in~\cite{BHMPP20,BHMP} that boundedness of mim-width is the underlying reason why
some specific hereditary graph classes, characterized by two forbidden induced subgraphs,
admit polynomial-time algorithms for a range of problems including {\sc  $k$-Colouring} and its generalization {\sc List $k$-Colouring} (the algorithms are
given in~\cite{CSZ,CGKP15,HKLSS10}).
Here we prove that the same holds for certain {\it superclasses of convex graphs} known in the literature.
Essentially all the known polynomial-time algorithms for such classes are obtained by reducing to the class of convex graphs. We show that our new approach
via mim-width simplifies the analysis, unifies the sporadic approaches
and explains the reductions to convex graphs.

\medskip
\noindent
{\bf Mim-width.}
A set of edges $M$ in a graph $G$ is a {\it matching} if no two edges of $M$ share an endpoint. A matching $M$ is {\it induced} if there is no edge in $G$ between vertices of different edges of $M$. Let $(A,\overline{A})$ be a partition of the vertex set of a graph $G$. Then $G[A,\overline{A}]$ denotes the bipartite subgraph of $G$ induced by the edges with one endpoint in $A$ and the other in $\overline{A}$.
Vatshelle~\cite{Va12} introduced the notion of \emph{maximum induced matching width}, also called mim-width. Mim-width measures
the extent to which it is possible to decompose a graph $G$ along certain vertex partitions $(A,\overline{A})$ such that the size of a maximum induced matching in $G[A,\overline{A}]$ is small. The kind of vertex partitions permitted stems from classical branch decompositions,
as we explain below.

A \textit{branch decomposition} for a graph~$G$ is a pair $(T, \delta)$, where $T$ is a subcubic tree and $\delta$ is a bijection from~$V(G)$ to the leaves of $T$.
Every edge $e \in E(T)$ partitions the leaves of $T$ into two classes, $L_e$ and $\overline{L_e}$, depending on which component of $T-e$ they belong to.
Hence, $e$ induces a partition $(A_e, \overline{A_e})$ of $V(G)$, where $\delta(A_e) = L_e$ and $\delta(\overline{A_e}) = \overline{L_e}$.
Let $\cutmim_{G}(A_{e}, \overline{A_{e}})$ be the size of a maximum induced matching in $G[A_{e}, \overline{A_{e}}]$.
Then the \emph{mim-width} $\mimw_{G}(T, \delta)$ of $(T, \delta)$ is the maximum value of $\cutmim_{G}(A_{e}, \overline{A_{e}})$ over all edges $e\in E(T)$. The \emph{mim-width} $\mimw(G)$ of $G$ is the minimum value of $\mimw_{G}(T, \delta)$ over all branch decompositions $(T, \delta)$ for $G$.
See Figure~\ref{f-example} for an example of a branch decomposition for a graph.

A {\it caterpillar} is a tree $T$ that contains a path~$P$, the {\it backbone} of $T$, such that every vertex not on $P$ has a neighbour on~$P$ (the tree in Figure~\ref{f-example}~(b) is an example of a caterpillar).
 A \textit{linear branch decomposition} for a graph~$G$ is a pair $(T, \delta)$, where $T$ is a caterpillar and $\delta$ is a bijection from~$V(G)$ to the leaves
 of~$T$. The \emph{linear mim-width} of $G$ is the minimum value of $\mimw_{G}(T, \delta)$ over all linear branch decompositions $(T, \delta)$ for $G$. Note that Figure~\ref{f-example} is in fact an example of a linear branch decomposition.

The problems of computing the mim-width and linear mim-width of a graph are \NP-hard~\cite{SV16}. Moreover, approximating the mim-width in polynomial time within a constant factor of the optimal is not possible unless $\mathsf{NP} = \mathsf{ZPP}$~\cite{SV16}.
It is not known how to compute in polynomial time a branch decomposition for a graph~$G$ whose mim-width is bounded by some function in
$\mimw(G)$.
Nevertheless, if $G$ is from some graph class~${\cal G}$ of bounded mim-width, then this is often possible. In that case,
the mim-width of ${\cal G}$ is said to be {\it quickly computable}. One can then try to develop a polynomial-time algorithm for the graph problem under consideration via dynamic programming over the computed branch decomposition. We
give examples of such problems later.

\begin{figure}[t]
\hspace*{-0.5cm}
\begin{subfigure}[b]{0.4\textwidth}

\captionsetup{justification=raggedright}

\includegraphics[scale=0.8]{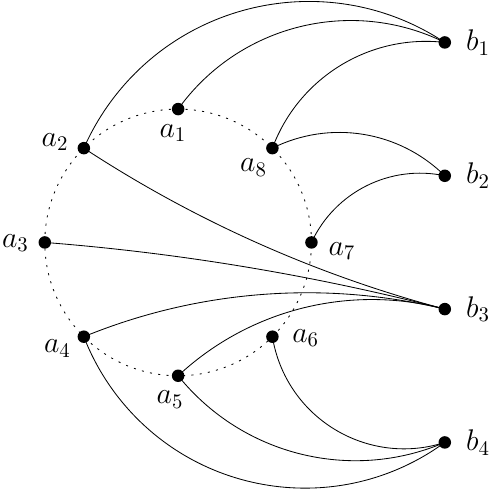}

\caption{}

\label{fig:circ}

\end{subfigure}
\hspace{1.3cm}
\begin{subfigure}[b]{0.55\textwidth}

\captionsetup{justification=raggedright}

\includegraphics[scale=0.5]{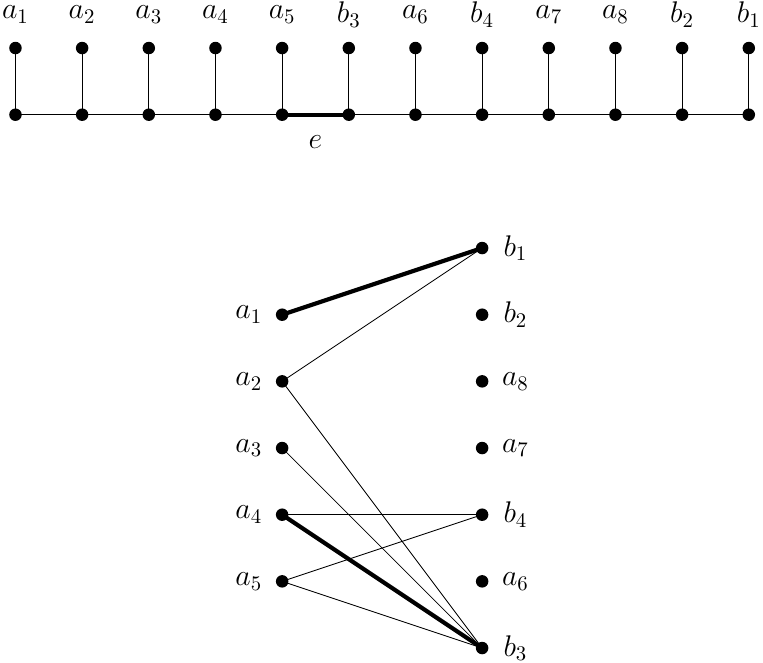}

\caption{}

\label{fig:branch}

\end{subfigure}
\caption{(\subref{fig:circ}) A circular convex graph $G = (A, B, E)$ with a circular ordering on $A$. (\subref{fig:branch}) A
(linear) branch decomposition $(T, \delta)$ for $G$, where $T$ is a caterpillar with a specified edge $e$, together with the graph $G[A_e, \overline{A_e}]$. The bold edges in $G[A_e, \overline{A_e}]$ form an induced matching and it is easy to see that $\cutmim_{G}(A_{e}, \overline{A_{e}}) = 2$.}\label{f-example}
\vspace*{-0.2cm}
\end{figure}

\smallskip
\noindent
{\bf Convex Graphs and Generalizations.}
A bipartite graph $G=(A,B,E)$ is {\it convex} if there exists a path $P$ with $V(P)=A$ such that the neighbours in $A$ of each $b\in B$ induce a connected subpath of $P$. Convex graphs generalize bipartite permutation graphs
(see, e.g., \cite{BLS99})
and form a well-studied graph class.
They were introduced in the sixties, by Glover \cite{Glo67}, to solve a special type of matching problem arising in some industrial application. Another early paper solving matching problems on convex graphs is by Lipski and Preparata \cite{LP81}.

Belmonte and Vatshelle \cite{BV13} proved that the mim-width of convex graphs is bounded and quickly computable.
 We consider superclasses of convex graphs and research to what extent mim-width can play a role in obtaining polynomial-time algorithms for problems on these classes.

Let ${\cal H}$ be a family of graphs.
A bipartite graph $G=(A,B,E)$ is ${\cal H}$-convex if there exists a graph $H\in {\cal H}$ with $V(H)=A$ such that the set of neighbours in $A$ of each $b\in B$ induces a connected subgraph of~$H$.
If ${\cal H}$ consists of all paths, we obtain the class of convex graphs.
Recall that a caterpillar is a tree $T$ that contains a
path~$P$ (the backbone), such that every vertex not on $P$ has a neighbour on~$P$. A caterpillar with a backbone consisting of one vertex is a {\it star}.
A {\it comb} is a caterpillar such that every backbone vertex has exactly one neighbour outside the backbone.
The {\it subdivision} of an edge $uv$ replaces $uv$ by a new vertex $w$ and edges $uw$ and $wu$. A {\it triad} is a tree that can be obtained from a $4$-vertex star after a sequence of subdivisions.
For $t, \Delta\geq 0$, a \textit{$(t, \Delta)$-tree} is a tree with maximum degree at most $\Delta$ and containing at most $t$ vertices of degree at least~$3$; note that, for example, a triad is a $(1,3)$-tree.
If ${\cal H}$ consists of all cycles, all trees, all stars, all triads, all combs or all $(t,\Delta)$-trees, then we obtain the class of \emph{circular convex graphs}, \emph{tree convex graphs}, \emph{star convex graphs}, \emph{triad convex graphs}, \emph{comb convex graphs} or \emph{$(t,\Delta)$-tree convex graphs}, respectively. See Figure~\ref{f-example} for an example of a circular convex graph (this class was introduced by Liang and Blum \cite{LB95} to model certain scheduling problems).

To show the relationships between the above graph classes we need some extra terminology.
Let $\mathcal{C}_{t, \Delta}$~be the class of $(t, \Delta)$-tree convex graphs. For fixed $t$ or $\Delta$, we have increasing sequences $\mathcal{C}_{t, 0} \subseteq \mathcal{C}_{t, 1} \subseteq \cdots$ and $\mathcal{C}_{0, \Delta} \subseteq \mathcal{C}_{1, \Delta} \subseteq \cdots$. For $t \in \mathbb{N}$, the class of \textit{$(t, \infty)$-tree convex graphs} is $\bigcup_{\Delta \in \mathbb{N}}\mathcal{C}_{t, \Delta}$, denoted by $\mathcal{C}_{t, \infty}$. Similarly, for $\Delta \in \mathbb{N}$, the class of \textit{$(\infty, \Delta)$-tree convex graphs} is $\bigcup_{t \in \mathbb{N}}\mathcal{C}_{t, \Delta}$, denoted by $\mathcal{C}_{\infty, \Delta}$.
Hence,
$\mathcal{C}_{t, \infty}$ and $\mathcal{C}_{\infty, \Delta}$ are the set-theoretic limits of the increasing sequences $\{\mathcal{C}_{t, \Delta}\}_{\Delta \in \mathbb{N}}$ and $\{\mathcal{C}_{t, \Delta}\}_{t \in \mathbb{N}}$, respectively. The class of \textit{$(\infty, \infty)$-tree convex graphs} is $\bigcup_{t, \Delta \in \mathbb{N}}\mathcal{C}_{t, \Delta}$, which coincides with the class of tree convex graphs. Notice that the class of convex graphs coincides with $\mathcal{C}_{t, 2}$, for any $t \in \mathbb{N} \cup \{\infty\}$, and with $\mathcal{C}_{0, \Delta}$, for any $\Delta \in \mathbb{N} \cup \{\infty\}$. The class of star convex graphs coincides with $\mathcal{C}_{1, \infty}$. Moreover, each triad convex graph belongs to~$\mathcal{C}_{1, 3}$ and each comb convex graph belongs to $\mathcal{C}_{\infty, 3}$.

A bipartite graph is {\it chordal bipartite} if every induced cycle in it has exactly four vertices.
Every convex graph is chordal bipartite (see, e.g., \cite{BLS99}) and every chordal bipartite graph is tree convex (see~\cite{JLWX13,Liu14}).
In Figure~\ref{fig:dia} we display the relationships between these classes.

Brault{-}Baron et al.~\cite{BCM15} proved that chordal bipartite graphs have unbounded mim-width. Hence, the result of~\cite{BV13} for convex graphs cannot be generalized to chordal bipartite graphs. We determine the mim-width of the other classes in Figure~\ref{fig:dia} but first discuss known algorithmic results for these classes.

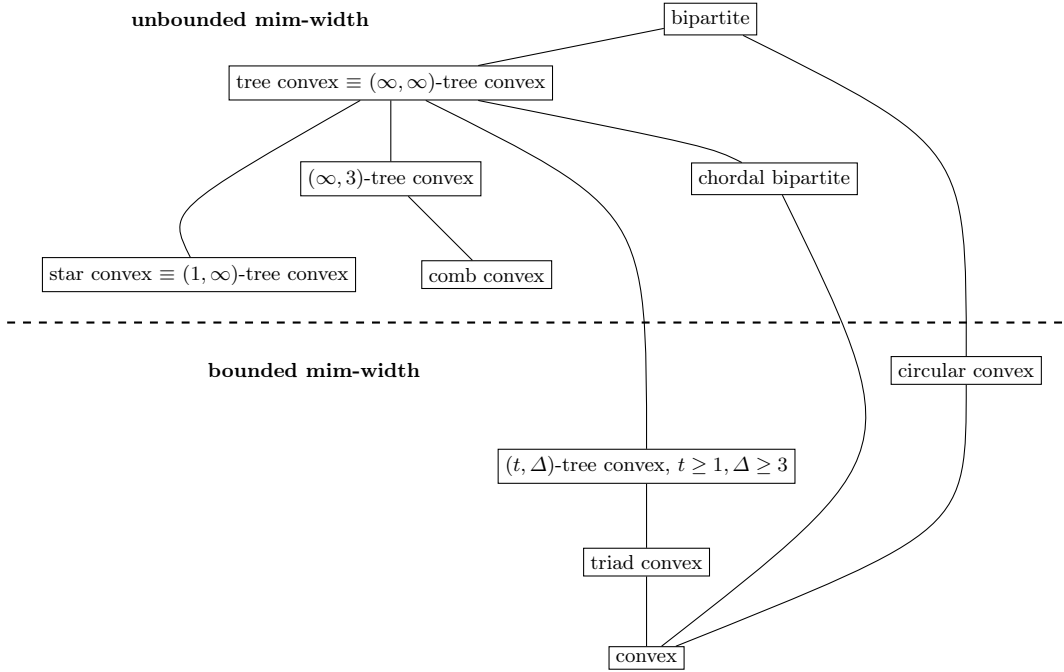
\begin{figure}
\centering
\begin{tikzpicture}[scale=0.85, every node/.style={transform shape}]
\node[rectangle,draw] (bipartite) at (6,16) {bipartite};
\node[rectangle,draw] (treeconvex) at (1,15) {tree convex $\equiv$ $(\infty, \infty)$-tree convex};
\node[rectangle,draw] (circularconvex) at (10,10.5) {circular convex};
\node[rectangle,draw] (tr) at (5,9) {$(t, \Delta)$-tree convex, $t\geq 1, \Delta \geq 3$};
\node[rectangle,draw] (inftyr) at (1,13.5) {$(\infty, 3)$-tree convex};
\node[rectangle,draw] (combconvex) at (2.5,12) {comb convex};
\node[rectangle,draw] (starconvex) at (-2,12) {star convex $\equiv$ $(1, \infty)$-tree convex}; %(0,10)
\node[rectangle,draw] (chordalbip) at (7,13.5) {chordal bipartite};
\node[rectangle,draw] (triad) at (5,7.5) {triad convex};
\node[rectangle,draw] (convex) at (5,6) {convex};
\node (classI) at (-1.2,16) {{\bf unbounded mim-width}};
\node (classI) at (-0.2,10.5) {{\bf bounded mim-width}};

\draw[-] (bipartite) ..controls (10,14).. (circularconvex)
(bipartite) -- (treeconvex)
(treeconvex) ..controls (6,14)..  (chordalbip)
(treeconvex) ..controls (5,13).. (tr)
(starconvex) ..controls (-2.5,13).. (treeconvex)
(combconvex) -- (inftyr)
(inftyr) -- (treeconvex)
(tr) -- (triad)
(triad) -- (convex)
(convex) ..controls (9.1,9.2).. (chordalbip)
(convex) ..controls (10,8)..  (circularconvex);

\draw [dashed,thick] (-5,11.25) to (11.5,11.25);
\end{tikzpicture}
\caption{The inclusion relations between the classes we consider.
A line from a lower-level class to a higher one means the first class is contained in the second.}\label{fig:dia}
\vspace*{-0.13cm}
\end{figure}

\medskip
\noindent
{\bf Known Results.}
Belmonte and Vatshelle \cite{BV13} and Bui-Xuan et al.~\cite{BTV13} proved that so-called  Locally Checkable Vertex Subset and Vertex Partitioning (LC-VSVP) problems, first defined in~\cite{TP97},
are polynomial-time solvable on graph classes whose mim-width is bounded and quickly computable.
This result was extended by Bergougnoux and Kant\'e~\cite{BK19} to variants of such problems with additional constraints on connectivity or acyclicity.
Each of the problems mentioned below is a special case of a Locally Checkable Vertex Subset (LCVS) problem
possibly with one of the two extra constraints.
We refer to the listed papers for the definitions of the problems, as we do not need them here.

Panda et al.~\cite{PPCDK20} proved that {\sc Induced Matching} is polynomial-time solvable for
circular convex and triad convex graphs, but \NP-complete for star convex and comb convex graphs. Pandey and Panda \cite{PP19} proved that {\sc Dominating Set} is polynomial-time solvable for circular convex, triad convex and
$(1, \Delta)$-tree convex graphs for every $\Delta\geq 1$. Liu et al.~\cite{LLX15} proved that {\sc Connected Dominating Set} is polynomial-time solvable for  circular convex and triad convex graphs. Chen et al.~\cite{CLL16} showed that {\sc Dominating Set}, {\sc Connected Dominating Set}
 and {\sc Total Dominating Set} are
 \NP-complete for star convex and comb convex graphs.

Lu et al.~\cite{LLX13} proved that {\sc Independent Dominating Set} is polynomial-time solvable for circular convex and triad convex graphs. The latter result was previously shown in~\cite{SLX12} using a dynamic programming approach instead of a reduction to convex graphs~\cite{LLX13}. Song et al.~\cite{SLX12} showed in fact a stronger result, namely that {\sc Independent Dominating Set}
 is polynomial-time solvable for $(t,\Delta)$-tree convex graphs for every $t\geq 1$ and $\Delta\geq 3$. They also showed that {\sc Independent Dominating Set} is \NP-complete for star convex and comb convex graphs~\cite{SLX12}. Hence, they obtained a dichotomy: {\sc Independent Dominating Set} is polynomial-time solvable for $(t,\Delta)$-tree convex graphs for every $t\geq 1$ and $\Delta\geq 3$ but \NP-complete for $(\infty,3)$-tree convex graphs and $(1,\infty)$-tree convex graphs.

The same dichotomy (explicitly formulated in~\cite{WLJX12})
holds for {\sc Feedback Vertex Set} and is obtained similarly. Namely, Jiang et al.~\cite{JLX11} proved that this problem is polynomial-time solvable for triad convex graphs and mentioned that their algorithm can be generalized to $(t,\Delta)$-tree convex graphs for every $t\geq 1$ and $\Delta\geq 3$.
Jiang et al.~\cite{JLWX13} proved that {\sc Feedback Vertex Set} is \NP-complete for star convex and comb convex graphs.
In addition, Liu et al.~\cite{LLLX14} proved that {\sc Feedback Vertex Set} is polynomial-time solvable for circular convex graphs, whereas Jiang et al.~\cite{JLWX13} proved that the {\sc Weighted Feedback Vertex Set} problem is polynomial-time solvable for triad convex graphs.

It turns out that the above problems are polynomial-time solvable on circular convex graphs and subclasses of $(t,\Delta)$-tree convex graphs, but \NP-complete for star convex graphs and comb convex graphs.
In contrast, Panda and Chaudhary~\cite{PC19} proved that {\sc Dominating Induced Matching} is not only polynomial-time solvable on circular convex and triad convex graphs, but also on star convex graphs.
Nevertheless, we notice {\it a common pattern}: many dominating set, induced matching and graph transversal type of problems are polynomial-time solvable for $(t,\Delta)$-tree convex graphs, for every $t\geq 1$ and $\Delta\geq 3$, and \NP-complete for comb convex graphs, and thus for $(\infty,3)$-tree convex graphs, and star convex graphs, or equivalently, $(1,\infty)$-tree convex graphs. Moreover, essentially all the  polynomial-time algorithms reduce the input to a convex graph.

\medskip
\noindent
{\bf Our Results.}
We simplify the analysis, unify the above approaches and explain the reductions to convex graphs, using mim-width. We prove three results that, together with the fact that chordal bipartite graphs have unbounded mim-width~\cite{BCM15}, explain the dotted line in Figure~\ref{fig:dia}. The first two results generalize the result of \cite{BV13} for convex graphs. The third result gives two new reasons why tree convex graphs (that is, $(\infty,\infty)$-tree convex graphs) have unbounded mim-width.

\begin{theorem}\label{t-1}
The mim-width of the class of circular convex graphs is bounded and quickly computable.
\end{theorem}

\begin{theorem}\label{bounded}
For every $t, \Delta \in \mathbb{N}$ with $t \geq 1$ and $\Delta\geq 3$,
the mim-width of the class of  $(t, \Delta)$-tree convex graphs is bounded and quickly computable.
\end{theorem}

\begin{theorem}\label{t-3}
The class of star convex graphs and the class of comb convex graphs have unbounded mim-width.
\end{theorem}
Hence, we obtain a structural dichotomy (recall that star convex graphs are the $(1,\infty)$-tree convex graphs and that
comb convex graphs are $(\infty,3)$-tree convex):

\begin{corollary}\label{c-1} Let $t, \Delta \in \mathbb{N} \cup \{\infty\}$ with $t\geq 1$ and $\Delta\geq 3$.
The class of $(t, \Delta)$-tree convex graphs has bounded mim-width if and only if $\{t, \Delta\} \cap \{\infty\} = \varnothing$.
\end{corollary}

\noindent
{\bf Algorithmic Consequences.}
As discussed, the following six problems were shown to be $\mathsf{NP}$-complete for star convex and comb convex graphs, and thus for $(1,\infty)$-tree convex graphs and $(\infty,3)$-tree convex graphs:
\textsc{Feedback Vertex Set} \cite{BK19,JLWX13};
\textsc{Dominating Set}, \textsc{Connected Dominating Set}, \textsc{Total Dominating Set} \cite{CLL16};
\textsc{Independent Dominating Set} \cite{SLX12};
\textsc{Induced Matching}~\cite{PPCDK20}.
These problems are examples of LCVS problems,
possibly with connectivity or acyclicity constraints. Hence, they are polynomial-time solvable for every graph class whose mim-width is bounded and quickly
computable~\cite{BV13,BK19,BTV13}.
Recall that the same holds for {\sc Weighted Feedback Vertex Set}~\cite{JKT20} and {\sc (Weighted) Subset Feedback Vertex Set}~\cite{BPT20}; these three problems generalize {\sc Feedback Vertex Set} and are thus \NP-complete for
star convex graphs and comb convex graphs.
 Combining these results with Corollary~\ref{c-1} yields the following complexity dichotomy.

\begin{corollary}\label{c-2}
  Let $t,\Delta \in \mathbb{N} \cup \{\infty\}$ with
  $t\geq 1$, $\Delta\geq 3$ and $\Pi$ be one of the nine problems mentioned above,
 restricted to $(t, \Delta)$-tree convex graphs.
  If $\{t, \Delta\} \cap \{\infty\} = \varnothing$, then $\Pi$ is polynomial-time solvable; otherwise, $\Pi$ is \NP-complete.
\end{corollary}

\noindent
It is worth noting that this complexity dichotomy does not hold for all LCVS problems; recall that \textsc{Dominating Induced Matching} is polynomial-time solvable on star convex graphs~\cite{PC19}. Theorems~\ref{t-1} and~\ref{bounded}, combined with the result of~\cite{BTV13}, imply that this problem is also polynomial-time solvable on circular convex graphs and $(t,\Delta)$-tree convex graphs for every $t\geq 1$ and $\Delta\geq 3$.

\medskip
\noindent
Our results have further algorithmic consequences.
For every fixed integer~$k\geq 1$, the {\sc $k$-Colouring} problem is an example of an LC-VSVP problem. Kwon~\cite{Kw} observed that
even its generalization {\sc List $k$-Colouring} is polynomial-time solvable on graph classes whose mim-width is bounded and quickly computable (see~\cite{BHMP} for details). Hence, Theorems~\ref{t-1} and~\ref{bounded}, combined with Kwon's observation~\cite{Kw},
also generalize a result of D{\'{\i}}az et al.~\cite{DDSS20} for {\sc List $k$-Colouring} on convex graphs to
 circular convex and $(t,\Delta)$-tree convex graphs (for any fixed $t\geq 1$, $\Delta\geq 3$).
 We prove the following result (note that the complexity of {\sc List $3$-Colouring} for comb convex graphs is still unresolved).

 \begin{theorem}\label{t-lc}
For $k\geq 4$, {\sc List $k$-Colouring} is \NP-complete for star convex graphs and comb convex graphs,
while {\sc List $3$-Colouring} is polynomial-time solvable for star convex graphs.
 \end{theorem}

 \noindent
 Thus, in addition to Corollary~\ref{c-2}, we obtain the following corollary for yet another problem.

 \begin{corollary}\label{c-22}
 For $k\geq 4$ and $t,\Delta \in \mathbb{N} \cup \{\infty\}$ with
  $t\geq 1$, $\Delta\geq 3$, {\sc List $k$-Colouring},
 restricted to $(t, \Delta)$-tree convex graphs, is polynomial-time solvable
 if $\{t, \Delta\} \cap \{\infty\} = \varnothing$, and \NP-complete otherwise.
\end{corollary}

\noindent
{\bf Outline.}
We prove Theorems~\ref{t-1}--\ref{t-3} in Sections~\ref{s-1}--\ref{s-3}, respectively, by showing additional structural results.
First, in Section~\ref{s-1}, we prove Theorem~\ref{t-1} by showing a similar but stronger result for linear mim-width of circular convex graphs.
Next, in Section~\ref{s-2}, we prove a stronger version of Theorem~\ref{bounded}; namely we prove the result for thinness, a relatively recent graph width parameter introduced in~\cite{M-O-R-C-thinness}.
This replaces the more direct proof given in~\cite{BBMP21}.
Finally, in Section~\ref{s-3}, we prove that Theorem~\ref{t-3} follows from a stronger result as well: we prove that star convex graphs and comb convex graphs have unbounded sim-width. Recall that the latter width parameter is more powerful than mim-width~\cite{KKST17}.

The above additional results naturally beg the question for a more refined analysis on width parameters for generalized convex classes.
We perform this analysis in Section~\ref{s-refined}. There, we consider a hierarchy of width parameters, which include mim-width and thinness. We determine exactly which of the generalized convex classes considered in the previous sections have bounded width for each of these parameters.
We are not yet aware of any new algorithmic results arising from this analysis,
but in Section~\ref{s-final} we give some directions for future research, after proving Theorem~\ref{t-lc} in Section~\ref{s-new2}.

\medskip
\noindent
{\bf Preliminaries.}
We consider only finite graphs $G=(V,E)$ with no loops and no multiple edges.
For $v \in V$, the \textit{neighbourhood} $N_G(v)$ is the set of vertices adjacent to $v$. The \textit{degree} $d(v)$ of a vertex $v \in V$ is the size $|N_G(v)|$. A vertex of degree $k$ is a \textit{$k$-vertex}. A graph is \textit{subcubic} if every vertex has degree at most~$3$. We let $\Delta(G) = \max\{d(v) : v \in V\}$.
For disjoint $S,T \subseteq V$, we say that $S$ is \textit{complete to} $T$ if every vertex of $S$ is adjacent to every vertex of $T$. For $S\subseteq V$, $G[S]=(S,\{uv\; :\; u,v\in S, uv\in E\})$
is the subgraph of $G$ \textit{induced} by $S$. The \textit{disjoint union} $G+H$ of graphs $G$ and $H$ has vertex set $V(G) \cup V(H)$ and edge set $E(G) \cup E(H)$.

A graph is \textit{$r$-partite}, for $r \geq 2$, if its vertex set admits a partition into $r$ classes such that every edge has its endpoints in different classes. A $2$-partite graph is also called \textit{bipartite}. A graph $G$ is a \textit{support} for a hypergraph $H = (V, \mathcal{S})$ if the vertices of $G$ correspond to the vertices of $H$ and, for each hyperedge $S \in \mathcal{S}$, the subgraph of $G$ induced by $S$ is connected.
When a bipartite graph $G=(A,B,E)$ is viewed as a hypergraph $H=(A,\{N(b) : b \in B\})$, then a support $T$ for $H$ with $T \in \mathcal{H}$ is a witness that $G$ is $\mathcal{H}$-convex.

\section{The Proof of Theorem~\ref{t-1}}\label{s-1}

We need the following known lemma on recognizing circular convex graphs.

\begin{lemma}[see, e.g., Buchin et al.~\cite{BKM11}]\label{reccirc} Circular convex graphs can be recognized and a cycle support computed, if it exists, in polynomial time.
\end{lemma}

We now define a special type of caterpillar.
For an integer $\ell \geq 1$, an {\em $\ell$-caterpillar} is a subcubic tree $T$ on
$2\ell$ vertices with $V(T) = \{ s_1,\ldots,s_\ell,t_1,\ldots,t_\ell \}$,
such that $E(T) = \{ s_it_i \;:\; 1 \leq i \leq \ell \} \cup \{ s_is_{i+1} \;:\; 1 \leq i \leq \ell-1 \}$.
Note that we label the leaves of an $\ell$-caterpillar $t_1,t_2,\dotsc,t_\ell$, in this order. Given a total ordering~$\prec$ of length $\ell$, we say that $(T,\delta)$ is \emph{obtained from $\prec$} if $T$ is an $\ell$-caterpillar and $\delta$ is the natural bijection from the $\ell$ ordered elements to the leaves of $T$.
We can now state our first result.

\begin{theorem}\label{t-ll}
Every circular convex graph $G$ has linear mim-width at most~$2$.
Moreover, we can construct a linear branch decomposition $(T,\delta)$ for $G$ with $\mimw_G(T,\delta) \le 2$ in polynomial time.
\end{theorem}

\begin{proof} Let $G = (A, B, E)$ be a circular convex graph with a circular ordering on $A$. By \Cref{reccirc}, we construct in polynomial time such an ordering $a_{1}, \dots, a_{n}$, where $n = |A|$ (see \Cref{f-example}). Let
$B_1=N(a_n)$
and $B_{2} = B \setminus B_{1}$. We obtain a total ordering $\prec$ on $V(G)$ by extending the ordering $a_{1}, \dots, a_{n}$ as follows. Each $b \in B_{1}$ is inserted after $a_{n}$, breaking ties arbitrarily. Each $b \in B_{2}$ is inserted immediately after the largest element of $A$ it is adjacent to (hence immediately after some $a_{i}$ with $1 \leq i < n$), breaking ties arbitrarily.

Let $T$ be the $|V(G)|$-caterpillar obtained from $\prec$. We will prove that $\mimw_G(T,\delta) \leq 2$. Let $e\in E(T)$.
We may assume without loss of generality that $e$ is not incident to a leaf of $T$.
Let $M$ be a maximum induced matching of $G[A_e, \overline{A_e}]$. As $e$ is not incident to a leaf, we may assume without loss of generality that each vertex in $\overline{A_e}$ is larger than any vertex in $A_e$ in the ordering $\prec$.

We first observe that at most one edge of $M$ has one endpoint in $B_{2}$. Indeed, suppose there exist two edges $xy, x'y' \in M$, each with one endpoint in $B_{2}$, say without loss of generality $\{y, y'\} \subseteq B_{2}$. Since each vertex in $B_{2}$ is adjacent only to smaller vertices, $\{y, y'\} \subseteq \overline{A_e}$ and $\{x, x'\} \subseteq A_e$. Without loss of generality, $y \prec y'$. However, $N(y)$ and $N(y')$ are intervals of the ordering and so either $x \in N(y')$ or $x' \in N(y)$, contradicting the fact that $M$ is induced.

We now show that at most two edges in $M$ have an endpoint in $B_{1}$ and, if exactly two such edges are in $M$, then no edge with an endpoint in $B_{2}$ is. First suppose that three edges of $M$ have one endpoint in $B_{1}$ and let $u_{1}, u_{2}, u_{3}$ be these endpoints. Since $N(u_{1})$, $N(u_{2})$ and $N(u_{3})$ are intervals of the circular ordering on $A$ all containing $a_{n}$, one of these neighbourhoods is contained in the union of the other two, contradicting the fact that $M$ is induced.

Finally, suppose that exactly two edges $u_{1}v_{1}$ and $u_{2}v_{2} \in M$ have one endpoint in $B_{1}$ and thus their other endpoint in $A$. Let $\{u_1, u_2\} \subseteq \overline{A_e}$ and $\{v_1, v_2\} \subseteq A_e$. Then, as each vertex in $\overline{A_e}$ is larger than any vertex in $A_e$ in $\prec$, we find that $u_1$ and $u_2$ belong to $B_1$ and thus $\{v_1, v_2\} \subseteq A$.
Now if there is some edge $u_3v_3 \in M$ such that $u_3 \in B_2$, then $u_3 \in \overline{A_e}$. Recall that $N(u_1)$ and $N(u_{2})$ are intervals of the circular ordering on $A$ both containing $a_{n}$. Since $M$ is induced, for each $i, j \in \{1, 2\}$, we have that $v_{i} \in N(u_{j})$, if $i = j$, and $v_{i} \notin N(u_{j})$, if $i \neq j$. This implies that one of $v_{1}$ and $v_{2}$ is larger than $v_{3}$ in $\prec$ and so it is contained in~$N(u_3)$, contradicting the fact that $M$ is induced. This concludes the proof. \qed
\end{proof}

\noindent
The proof of Theorem~\ref{t-1} is now an easy consequence of the above result.

\medskip
\noindent
\faketheorem{Theorem~\ref{t-1} (restated).}
{\it The mim-width of the class of circular convex graphs is bounded and quickly computable.}

\begin{proof}
By definition, the mim-width of a graph $G$ is at most the linear mim-width of $G$. Hence, the result follows from Theorem~\ref{t-ll}. \qed
\end{proof}

\section{The Proof of Theorem~\ref{bounded}}\label{s-2}

We need the following lemma on recognizing $(t, \Delta)$-tree convex graphs\footnote{Jiang et al.~\cite{JLWX13} proved that {\sc Weighted Feedback Vertex Set} is polynomial-time solvable for triad convex graphs if a triad support is given as input. They observed that an associated tree support can be constructed in linear time, but this does not imply that a triad support can be obtained. \Cref{rectr} shows that indeed a triad support can be obtained in polynomial time and need not be provided on input.}.

\begin{lemma}\label{rectr} For $t, \Delta \in \mathbb{N}$, $(t, \Delta)$-tree convex graphs can be recognized and a $(t, \Delta)$-tree support computed, if it exists, in
$O(n^{t+3})$
time.
\end{lemma}

\begin{proof} Given a hypergraph $H = (V, \mathcal{S})$ together with degrees $d_{i}$ for each $i \in V$, Buchin et al.~\cite{BKM11} provided an $O(|V|^{3} + |\mathcal{S}||V|^{2})$ time algorithm that solves the following decision problem: Is there a tree support for $H$ such that each
vertex~$i$ of the tree has degree at most $d_{i}$? If it exists, the algorithm computes a tree support satisfying this property. Given as input a bipartite graph $G = (A, B, E)$, we consider the hypergraph $H = (A, \mathcal{S})$, where $\mathcal{S} = \{N(b) : b \in B\}$. For each of the ${|A|\choose t} = O(|A|^{t})$ subsets $A' \subseteq A$ of size~$t$ we proceed as follows: we assign a degree~$\Delta$ to each of its elements and a degree~$2$ to each element in $A \setminus A'$. We then apply the algorithm in \cite{BKM11} to the $O(|A|^{t})$ instances thus constructed. If $G$ is $(t, \Delta)$-tree convex, then the algorithm returns a $(t, \Delta)$-tree support for $H$. \qed
\end{proof}

\begin{figure}[t]
\begin{center}
    \begin{tikzpicture}[scale=0.7]

\foreach \x in {0.5} {
    \vertex{-6}{0+4}{v1};
    \vertex{-6}{3*\x+4}{v2};
    \vertex{-6}{6*\x+4}{v3};
    \vertex{-6}{8*\x+4}{v4};

    \vertex{-4}{2*\x+4}{w1};
    \vertex{-4}{5*\x+4}{w2};
    \vertex{-4}{10*\x+4}{w3};

    \vertex{-2}{1*\x+4}{z1};
    \vertex{-2}{4*\x+4}{z2};
    \vertex{-2}{7*\x+4}{z3};
    \vertex{-2}{9*\x+4}{z4};
    \vertex{-2}{11*\x+4}{z5};
}

\node[label=above:{$V^1$}] at (-6,2) {};
\node[label=above:{$V^2$}] at (-4,2) {};
\node[label=above:{$V^3$}] at (-2,2) {};

\path (v4) edge [e1,bend right=25] (v2); \path (v4) edge [e1]
(v3);

\path (w3) edge [e1] (w2);

\path (z4) edge [e1,bend left=25] (z2); \path (z4) edge [e1] (z3);
\path (z3) edge [e1] (z2);

\path (v2) edge [e1] (w1); \path (v3) edge [e1] (w2); \path (v4)
edge [e1] (w1); \path (v4) edge [e1] (w2); \path (v1) edge [e1]
(w2); \path (v4) edge [e1] (w3); \path (v2) edge [e1] (w2); \path
(v1) edge [e1] (w1);

\path (v1) edge [e1] (z1); \path (v2) edge [e1] (z2); \path (v2)
edge [e1,bend left=15] (z1); \path (v1) edge [e1,bend left=15]
(z2); \path (v2) edge [e1,bend left=15] (z3); \path (v4) edge [e1]
(z4); \path (v3) edge [e1] (z4); \path (v3) edge [e1] (z2); \path
(v3) edge [e1] (z3);

\path (z2) edge [e1] (w1); \path (z3) edge [e1] (w2); \path (z5)
edge [e1] (w3); \path (z5) edge [e1] (w2); \path (z2) edge [e1]
(w2); \path (z1) edge [e1] (w1); \path (z1) edge [e1] (w2);

\end{tikzpicture}
\end{center}
\caption{A $3$-thin representation of a graph. The vertices are
ordered increasingly by their $y$-coordinate, and the classes
correspond to the vertical lines.}\label{fig:thinness}
\end{figure}
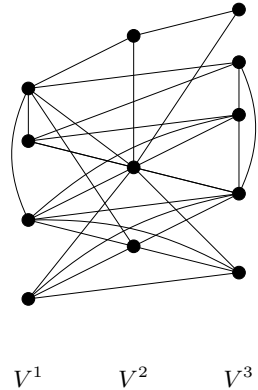

As mentioned in Section~\ref{sec:intro}, we will prove Theorem~\ref{bounded} as a consequence of a stronger result for a less powerful width parameter, namely thinness.
A graph $G=(V,E)$ is \emph{$k$-thin} if there exists an ordering
$v_1, \dots , v_n$ of~$V$ and a partition of~$V$ into $k$ classes
$(V^1,\ldots,V^k)$ such that, for every triple $(r,s,t)$ with
$r<s<t$, the following holds: if $v_r$ and $v_s$ belong to the same
class $V^i$ for some $i\in\{1,\ldots,k\}$ and moreover $v_r v_t \in E$, then we have $v_sv_t \in E$. In this case, the ordering
$v_1, \dots , v_n$ and the partition $(V^1,\dots,V^k)$ are said to be \emph{consistent} and form a {\it $k$-thin representation} for $G$.
 The \emph{thinness} $\thin(G)$ of $G$ is the minimum integer~$k$ such that $G$ is $k$-thin. See Figure~\ref{fig:thinness} for an example of a $3$-thin graph.

The following theorem shows that a class of $(t,\Delta)$-tree convex graphs, for some $t \ge 1$ and $\Delta \ge 3$, has bounded thinness.

\begin{theorem}\label{t-thin}
For every $t, \Delta \in \mathbb{N}$ with $t \geq 1$ and $\Delta\geq 3$,
the class of $(t, \Delta)$-tree convex graphs has thinness at most $2+t(\Delta-2)$, and a corresponding
$(2+t(\Delta-2))$-thinness ordering and partition can be computed in polynomial time.
\end{theorem}

\begin{proof}
Let $G=(A,B,E)$ be a $(t, \Delta)$-tree convex graph for some $t\geq 1$ and $\Delta\geq 3$. Let $V^1 = B$. We can compute a $(t,\Delta)$-tree support in polynomial time by Lemma~\ref{rectr}.
We root this tree at an arbitrary leaf
and assign the root to the class $V^2$. If a node has one child,
then it is assigned to the same class as its parent. If it has
more than one child, then one child is assigned to the same class
as its parent, and each of the other children starts a new
class. So, we have at most $2+t(\Delta-2)$ classes (including $V^1$).

We order the vertices of $A$ by postorder of the support tree, and
insert each vertex $b$ of $B$ into the order in such a way that
the greatest vertex of $A$ smaller than it is the greatest
neighbour of $b$ in $A$. Note that each class of vertices
of $A$ is a linearly ordered path of the support tree.

We will now show that the order and the partition are consistent. The vertices
of $B$ are an independent set, the vertices of $A$ are an
independent set, and no vertex of $A$ has a neighbour in $B$
smaller than it. So, suppose that $a < a' < b$, with $a, a' \in V^j$ for
some $j \geq 2$, $b \in V^1$, and $ab \in E$.

Let $a''$ be the greatest neighbour of $b$ in $A$. By construction of the order, $a''$ is the greatest vertex of~$A$ that is smaller than $b$.
Hence, $a < a' \leq a'' < b$ (note that $a'=a''$ is possible). Moreover, $ab
\in E$ and $a''b \in E$. If $a'' \in V^j$, then clearly $a'$
belongs to the unique path joining $a$ and $a''$ in the support
tree, so $a'b \in E$. If $a'' \in V^{j'}$, for $j'\neq j$, then the
common ancestor of $a$ and $a''$ in the support tree is also a
neighbour of $b$. Indeed, since we have ordered the tree by
postorder, the common ancestor is precisely $a''$. Also in this
case, $a'$ belongs to the unique path joining $a$ and $a''$ in the
support tree, so $a'b \in E$.

From the above we conclude that $G$ has thinness at most $2+t(\Delta-2)$. Moreover,  we can construct the
$(2+t(\Delta-2))$-thinness ordering and partition above in polynomial time. \qed
\end{proof}

\noindent
We also need the following known result, which shows that linear mim-width (and thus mim-width) is a more powerful width parameter than thinness.

\begin{lemma}[\cite{BE19}]\label{l-mwt}
For every graph $G$, the linear mim-width of $G$ is at most the thinness of $G$. Moreover, from every $k$-thin representation for $G$ we can obtain, in linear time, a linear branch decomposition for $G$.
\end{lemma}

\noindent
We are now ready to prove Theorem~\ref{bounded}.

\medskip
\noindent
\faketheorem{Theorem~\ref{bounded} (restated).}
{\it  For every $t, \Delta \in \mathbb{N}$ with $t \geq 1$ and $\Delta\geq 3$,
the mim-width of the class of  $(t, \Delta)$-tree convex graphs is bounded and quickly computable.}

\begin{proof}
By definition, the mim-width of a graph $G$ is at most the linear mim-width of~$G$.
Combining Theorem~\ref{t-thin} with Lemma~\ref{l-mwt} immediately yields that for every $t, \Delta \in \mathbb{N}$ with $t \geq 1$ and $\Delta\geq 3$,
the class of $(t, \Delta)$-tree convex graphs has (linear) mim-width at most $2+t(\Delta-2)$,
and that a corresponding branch decomposition can be computed in polynomial time. \qed
\end{proof}

\section{The Proof of Theorem~\ref{t-3}}\label{s-3}

For proving Theorem~\ref{t-3}, we need the following lemma.

\begin{lemma}[see Wang et al.~\cite{WCLT14}]\label{constr} Let $G = (A, B, E)$ be a bipartite graph and $G'$ be the bipartite graph obtained from $G$ by making $k$ new vertices complete to $B$. If $k = 1$, then $G'$ is star convex. If $k = |A|$, then $G'$ is comb~convex.
\end{lemma}

\noindent
As mentioned in Section~\ref{sec:intro}, we prove Theorem~\ref{t-3} as a consequence of as stronger result, which requires some extra terminology given below.

Consider a branch decomposition $(T,\delta)$ for a graph~$G$.
Recall that every edge $e \in E(T)$ partitions the leaves of $T$ into the classes $L_e$ and $\overline{L_e}$ and that
$e$ induces a partition $(A_e, \overline{A_e})$ of $V(G)$, where $\delta(A_e) = L_e$ and $\delta(\overline{A_e}) = \overline{L_e}$.
Let $\cutsim_{G}(A_{e}, \overline{A_{e}})$ be the size of a maximum induced matching $M$ in $G$ such that every edge of $M$ has one end-vertex in $A_e$ and the other end-vertex in $\overline{A_e}$.
Then the \emph{sim-width} $\simw_{G}(T, \delta)$ of $(T, \delta)$ is the maximum value of $\cutsim_{G}(A_{e}, \overline{A_{e}})$ over all edges $e\in E(T)$. The \emph{sim-width} $\simw(G)$ of $G$ is the minimum value of $\simw_{G}(T, \delta)$ over all branch decompositions $(T, \delta)$ for $G$.

We now prove the following result.

\begin{theorem}\label{t-simw}
The class of star convex graphs and the class of comb convex graphs have unbounded sim-width.
\end{theorem}

\begin{proof}
We use the fact that grids are bipartite and have unbounded sim-width~\cite{KKST17}.
Moreover, just as for mim-width, it is readily seen that adding a vertex to a graph does not decrease the sim-width.
Hence, we can apply \Cref{constr}: for any grid on partition classes $A$ and $B$, by adding a vertex complete to $B$ we obtain a star convex graph, and by adding $|A|$ new vertices complete to $B$ we obtain a comb convex graph.
Thus the class of star convex graphs and the class of comb convex graphs have unbounded sim-width. \qed
\end{proof}

\noindent
We also need the following observation of Kang, Kwon, Str{\o}mme and Telle.

\begin{lemma}[\cite{KKST17}]\label{l-kk}
For every graph $G$, the sim-width of $G$ is at most the mim-width of $G$.
\end{lemma}

\noindent
We can now prove Theorem~\ref{t-3}.

\medskip
\noindent
\faketheorem{Theorem~\ref{t-3} (restated).}
{\it The class of star convex graphs and the class of comb convex graphs have unbounded mim-width.}

\begin{proof}
The result immediately follows from Theorem~\ref{t-simw} after applying Lemma~\ref{l-kk}. \qed
\end{proof}

\begin{figure}
\begin{center}
\begin{tikzpicture}[yscale=.7]
\node at (0,10) (p) {sim-width};
\node at (0.4,9.4) (ref1) {\cite{KKST17}};
\node at (-3,8.5) (b) {};
\node at (0,8.5) (wc) {mim-width};
\node at (0.4,7.9) (ref2) {\cite{Va12}};
\node at (-3,7) (pe) {linear mim-width};
\node at (-6,7) (cb) {};
\node at (0,7) (c) {clique-width};
\node at (0.4,6.4) (ref3) {\cite{CO00}};
\node at (-2.7,6.4) (ref4) {\cite{BE19}};
\node at (-2.54,3.1) (ref5) {{\small $[$Section~\ref{s-new}$]$}};
\node at (3,7) (dh) {};
\node at (-6,5.5) (bp) {};
\node at (-1.5,5.5) (s) {};
\node at (0,5.5) (i) {treewidth};
\node at (0,2.5) (ui) {path-width};
\node at (-3,5.5) (spe) {thinness};
\node at (-3,4) (pt) {proper thinness};
\draw [->](pe) -- (spe);
\draw [->] (spe) -- (pt);
\draw [->] (pt) -- (ui);
\draw [->] (wc) -- (pe);
\draw [->] (p) -- (wc);
\draw [->] (wc) -- (c);
\draw [->] (c) -- (i);
\draw [->] (i) -- (ui);
\end{tikzpicture}
\end{center}
\caption{\label{f-power}
The relationships between the different width parameters that we consider in Section~\ref{s-refined}.
Parameter $p$ is more powerful than parameter $q$ if and only if there exists a directed path from $p$ to $q$. To explain the incomparabilities, proper interval graphs have proper thinness~1~\cite{BE19} and unbounded clique-width~\cite{GR00}, whereas trees have tree-width~1 and unbounded linear mim-width~\cite{HTV19}.
Unreferenced arrows follow from the definitions of the width parameters involved
except for the arrow from proper thinness to path-width whose proof we give in Section~\ref{s-new}.\\[15pt]}
\centering
\begin{tikzpicture}[scale=0.74, every node/.style={transform shape}]
\node[rectangle,draw] (bipartite) at (6,16) {bipartite};
\node[rectangle,draw] (treeconvex) at (1,15) {tree convex $\equiv$ $(\infty, \infty)$-tree convex};
\node[rectangle,draw] (circularconvex) at (10,10.5) {circular convex};
\node[rectangle,draw] (tr) at (5,9) {$(t, \Delta)$-tree convex, $t\geq 1, \Delta \geq 3$};
\node[rectangle,draw] (inftyr) at (1,13.5) {$(\infty, 3)$-tree convex};
\node[rectangle,draw] (combconvex) at (2.5,12) {comb convex};
\node[rectangle,draw] (starconvex) at (-2,12) {star convex $\equiv$ $(1, \infty)$-tree convex}; %(0,10)
\node[rectangle,draw] (chordalbip) at (7,13.5) {chordal bipartite};
\node[rectangle,draw] (triad) at (5,7.5) {triad convex};
\node[rectangle,draw] (convex) at (5,6) {convex};
\node (classI) at (-1.2,16) {{\bf unbounded sim-width}};
\node (classI) at (-0.2,10.5) {{\bf bounded linear mim-width, unbounded thinness}};
\node (classI) at (-1.2,8) {{\bf bounded thinness,}};
\node (classI) at (-1.2,7.5) {{\bf unbounded proper thinness,}};
\node (classI) at (-1.2,7) {{\bf unbounded clique-width}};

\draw[-] (bipartite) ..controls (10,14).. (circularconvex)
(bipartite) -- (treeconvex)
(treeconvex) ..controls (6,14)..  (chordalbip)
(treeconvex) ..controls (5,13).. (tr)
(starconvex) ..controls (-2.5,13).. (treeconvex)
(combconvex) -- (inftyr)
(inftyr) -- (treeconvex)
(tr) -- (triad)
(triad) -- (convex)
(convex) ..controls (9.1,9.2).. (chordalbip)
(convex) ..controls (10,8)..  (circularconvex);

\draw [dashed,thick] (-5,11.25) to (11.5,11.25)
(-5,9.75) to (11.5,9.75);
\end{tikzpicture}
\caption{The inclusion relations between the classes we consider.
A line from a lower-level class to a higher one means the first class is contained in the second.}
\label{fig:dia2}
\end{figure}

\section{A Refined Parameter Analysis}\label{s-refined}

We proved Theorems~\ref{t-1}--\ref{t-3} by showing stronger results for other width parameters.
In this section, we extend this more refined analysis on width parameters for the graph classes listed in Figure~\ref{fig:dia}. We will consider the graph width parameters listed in Figure~\ref{f-power}, which are path-width, treewidth, clique-width, (linear) mim-width, sim-width, thinness and proper thinness. The results in this section are summarized in Figure~\ref{fig:dia2}.
Note that we provide a {\it complete} picture with respect to the width parameters and graph classes considered.

We prove the three parts in Figure~\ref{fig:dia2} that are separated by the dotted lines in Sections~\ref{s-row1}--\ref{s-row3}, respectively; note that some results are shown already in the previous sections.
Afterwards, we prove in Section~\ref{s-new} the (only) unreferenced arrow in Figure~\ref{f-power}.

\subsection{Bounded Thinness but Unbounded Proper Thinness and Clique-Width}\label{s-row1}

In this section we consider the bottom row in Figure~\ref{fig:dia2}. In Section~\ref{s-2}, we already proved that for every $t, \Delta \in \mathbb{N}$ with $t \geq 1$ and $\Delta\geq 3$, the class of $(t, \Delta)$-tree convex graphs has bounded thinness (Theorem~\ref{t-thin}). Hence, it remains to show that convex graphs have unbounded proper thinness and unbounded clique-width.

A graph $G=(V,E)$ is \emph{proper $k$-thin} if there exists an ordering $v_1, \dots , v_n$ of~$V$ and a partition of~$V$ into $k$ classes $(V^1,\dots,V^k)$ such that for each triple $(r,s,t)$ with $r<s<t$ the following holds:
\begin{itemize}
\item if the vertices $v_r$ and $v_s$ are in the same class $V^i$ for some $i\in \{1,\ldots,k\}$ and $v_r v_t \in E$, then $v_sv_t \in E$; and
\item if the vertices $v_s$ and $v_t$ are in the same class $V^i$ for some $i\in \{1,\ldots,k\}$ and $v_rv_t\in E$, then $v_rv_s \in E$.
\end{itemize}
In this case, the ordering $v_1,\ldots,v_n$ and the partition $(V^1,\ldots,V^k)$ are \emph{strongly consistent}. The \emph{proper thinness} $\pthin(G)$ of $G$ is the minimum integer~$k$ such that $G$ is proper $k$-thin.
We cannot strengthen Theorem~\ref{t-thin} to proper thinness, nor to clique-width, due to the following result for convex graphs.

\begin{theorem}\label{proper}
The class of convex graphs has unbounded proper thinness and unbounded clique-width.
\end{theorem}

\begin{proof}
For the second part of the statement, we note that the class of bipartite permutation graphs, which form a subclass of convex graphs (see, e.g., \cite{BLS99}), has unbounded clique-width~\cite{BL03}. Hence, it remains to show the first part of the statement.

Given a vertex order $<$ and a subset $S \subseteq V(G)$, a set of
vertices $X \subseteq S$ is \emph{consecutive in $S$ according to
$<$} if there is no $z$ in $S \setminus X$ such that $\min(X) < z
< \max(X)$. Notice that for each vertex $u \in S$, there are at
most two vertices $x$ in $S$ such that $\{u, x\}$ is consecutive in
$S$ according to $<$. Namely, if $S = s_1 < \dots < s_r$ and $u =
s_i$, such vertices are $s_{i-1}$ when $i > 1$ and $s_{i+1}$ when
$i < r$.

We prove the following claim.

\medskip
\noindent
{\it Claim~1.} An order $<$ and a partition $V^1, \dots, V^k$ are strongly
consistent if and only if for every $v \in V(G)$ and every $1 \leq
j \leq k$, the set $N[v] \cap (V^j \cup \{v\})$ is consecutive in $V^j
\cup \{v\}$ according to $<$. In particular, $N[v] \cap V^j$ is
consecutive in $V^j$ according to $<$.

\medskip
\noindent We prove Claim~1 as follows. $\Rightarrow$) Let $v \in
V(G)$ and $1 \leq j \leq k$. Let $X = N[v] \cap (V^j \cup \{v\})$
and suppose there is some $z\in (V^j\cup \{v\})\setminus X$ (or
equivalently, $z\in V^j\setminus X$ as $v\in X$) such that
$\min(X) < z < \max(X)$. If $v < z$, then $v < z < \max(X)$; $z,
\max(X) \in V^j$; $v\max(X) \in E(G)$; and $vz \not \in E(G)$,
contradicting that the order and partition are strongly
consistent. If $v > z$, then $\min(X) < z < v$; $z, \min(X) \in
V^j$; $v\min(X) \in E(G)$; and $vz \not \in E(G)$, again
contradicting that the order and partition are strongly
consistent.

$\Leftarrow$) Let $r < s < t$ such that $rt \in E(G)$. Suppose
first $r,s \in V^j$ for some $1 \leq j \leq k$. Since the vertices
in $N[t] \cap (V^j \cup \{t\})$ are consecutive in $V^j \cup
\{t\}$ according to $<$, $r, t \in N[t] \cap (V^j \cup \{t\})$,
and $s \in V^j$, then $s \in N[t]$. The proof for $s,t \in V^j$
for some $1 \leq j \leq k$ is analogous by using the property for
$N[r]$.

As an immediate consequence, we have that $N[v] \cap V^j$ is
consecutive in $V^j$ according to~$<$. This completes the proof of Claim~1.

\medskip
\noindent
Let $\{G_k\}_{k\geq 1}$ be a family of bipartite convex graphs
defined recursively as follows: $G_1$ is the trivial graph, with the partition $(A,B)$ of $V(G_1)$ such that $|A|=1$ and $B = \varnothing$.
For $k \geq 2$, we define $G_k = (A,B)$ from the
disjoint union of three copies $H_i=(A_i,B_i)$, $i = 1,2,3$,  of
$G_{k-1}$, by adding a new vertex $u$ that we make complete to $A = A_1 \cup A_2
\cup A_3$ (thus, $B = B_1 \cup B_2 \cup B_3 \cup \{u\}$). Note
that for every $k \geq 2$, both $A$ and $B$ are nonempty. The graphs $G_1$, $G_2$ and $G_3$ are displayed in Figure~\ref{fig:thinness2}.

\begin{figure}
\begin{center}
    \begin{tikzpicture}[scale=0.7]

    \vertex{-1}{1}{v1};

    \vertex{2}{0}{w1};
    \vertex{2}{1}{w2};
    \vertex{2}{2}{w3};
    \vertex{4}{1}{w4};

    \vertex{6}{0.5}{z1};
    \vertex{6}{1}{z2};
    \vertex{6}{1.5}{z3};
    \vertex{8}{1}{z4};

    \vertex{6}{0}{x1};
    \vertex{6}{-0.5}{x2};
    \vertex{6}{-1}{x3};
    \vertex{8}{-0.5}{x4};

    \vertex{6}{2}{y1};
    \vertex{6}{2.5}{y2};
    \vertex{6}{3}{y3};
    \vertex{8}{2.5}{y4};

    \vertex{8}{1.75}{t0};

\node[label=above:{$G_1$}] at (-1,-3) {};
\node[label=above:{$G_2$}] at (3,-3) {};
\node[label=above:{$G_3$}] at (7,-3) {};

\path (w4) edge [e1] (w1); \path (w4) edge [e1] (w2); \path (w4)
edge [e1] (w3);

\path (z4) edge [e1] (z1); \path (z4) edge [e1] (z2); \path (z4)
edge [e1] (z3);

\path (x4) edge [e1] (x1); \path (x4) edge [e1] (x2); \path (x4)
edge [e1] (x3);

\path (y4) edge [e1] (y1); \path (y4) edge [e1] (y2); \path (y4)
edge [e1] (y3);

\path (t0) edge [e1] (z1); \path (t0) edge [e1] (z2); \path (t0)
edge [e1] (z3);

\path (t0) edge [e1] (x1); \path (t0) edge [e1] (x2); \path (t0)
edge [e1] (x3);

\path (t0) edge [e1] (y1); \path (t0) edge [e1] (y2); \path (t0)
edge [e1] (y3);

\end{tikzpicture}
\end{center}
\caption{The graphs $G_1$, $G_2$, $G_3$, from the family of graphs $\{G_k\}_{k\geq 1}$ in the proof of
Theorem~\ref{proper}. All the graphs in this family are convex and
their proper thinness increases with $k$.}\label{fig:thinness2}
\end{figure}
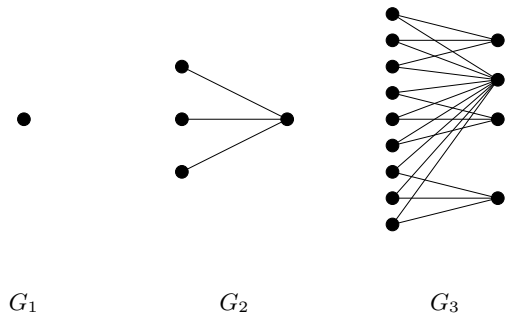

It remains to show the following claim.

\medskip
\noindent
{\it Claim 2.}
For every $k \geq 1$, $\pthin(G_k) = k$.

\medskip
\noindent
We first prove by induction that $\pthin(G_k)\geq k$.
For $k=1$, the statement is
true, as $G_1$ is the trivial graph.

Let $k \geq 2$. The graph
$G_k = (A,B)$ arises from the disjoint union of three copies
$H_i=(A_i,B_i)$, $i = 1,2,3$,  of $G_{k-1}$, by adding a vertex
$u$ complete to $A = A_1 \cup A_2 \cup A_3$ (thus $B = B_1 \cup
B_2 \cup B_3 \cup \{u\}$). Notice that for $v \in A_i$, $w \in
A_j$, $i \neq j$, it holds that $N[v] \cap N[w] = \{u\}$. On the other
hand, every vertex of $H_i$ belongs to $N[v]$ for some $v \in
A_i$.

Now let $<$ be an ordering of $V(G_k)$ and let $(V^1,\ldots,V^r)$, for some integer $r$, be a partition of $V(G_k)$, such that
 $<$ and $(V^1, \dots, V^r)$ are strongly
consistent. Let $j$ be such that $u \in V^j$. There are at
most two vertices $x$ in $V^j$ such that $\{u, x\}$ is consecutive in
$V^j$ according to $<$, so
there is at
least one value $i$, $1 \leq i \leq 3$, such that none of them belongs to $H_i$.
We claim that $V^j \cap V(H_i)
= \varnothing$. Suppose, on the contrary, that there exists a vertex $w$ in $V(H_i) \cap V^j$, and let $v \in A_i$ such that $w \in N[v]$.
Then $N[v]$ in $V^j$, which contradicts Claim~1.

Finally, since $<$ and the partition restricted to $V(H_i)$
are also strongly consistent, and, by the inductive hypothesis,
$\pthin(H_i) = k-1$, it follows that $r \geq k$, so $\pthin(G_k)
\geq k$.

We now prove that $\pthin(G_k) \leq k$. In order to do this, we will inductively build
an ordering and a partition of $V(G_k)$ into $k$ classes that are
strongly consistent, and such that $V^1 = A$. For $k=1$, this is
straightforward.
For $k \geq 2$, by induction, we can find an ordering and a partition into $k-1$ classes for the vertex set of each $H_i \cong G_{k-1}$, such that the ordering and partition are strongly consistent.  We concatenate the orderings, and take the union of the corresponding classes, to obtain an ordering and partition for $G_k - u$.  In particular, $V^1 = A_1 \cup A_2 \cup A_3$.
Finally, we create a new class
$V^k = \{u\}$, and make $u$ the greatest vertex in the order. It
is readily seen that the order and the partition that we created are strongly
consistent for~$G_k$. \qed
\end{proof}

\subsection{Bounded Linear Mim-Width but Unbounded Thinness}\label{s-row2}

By Theorem~\ref{t-ll}, the class of circular convex graphs has bounded linear mim-width. Below we prove that circular convex graphs have unbounded thinness. Hence, for circular convex graphs we cannot obtain the same result as for $(t,\Delta)$-tree convex graphs in Theorem~\ref{t-thin}.

\begin{theorem}\label{circular2}
The class of circular convex graphs has unbounded thinness.
\end{theorem}

\begin{proof}
The {\it crown} $H_n$ is the graph on $2n$ vertices that is obtained from a complete bipartite graph after removing a perfect matching. The class of crown graphs has unbounded thinness~\cite{BGOSS} and is readily seen to be circular convex. \qed
\end{proof}

\subsection{Unbounded Sim-Width}\label{s-row3}

From Theorem~\ref{t-simw} we know that both the class of star convex graphs and the class of comb convex graphs have unbounded sim-width.
Hence, it remains to prove that chordal bipartite graphs have unbounded sim-width, which we do below.
For two graphs $H_1$ and $H_2$, a graph $G$ is \textit{$(H_1,H_2)$-free} if $G$ has no induced subgraph isomorphic to $H_1$ or $H_2$.

\begin{theorem}\label{t-chordalbip}
The class of chordal bipartite graphs has unbounded sim-width.
\end{theorem}

\begin{proof}
Let $K_{3}\boxminus S_{3}$ be the graph that consists of a triangle on vertices $x,y,z$ to which we add three new vertices $x',y',z'$ with edges $xx'$, $yy'$ and $zz'$. Let $K_{3}\boxminus K_{3}$ be the graph obtained from two triangles on vertices $a_1,b_1,c_1$ and $a_2,b_2,c_2$, respectively, to which we add the edges $a_1a_2$, $b_1b_2$ and $c_1c_2$. Kang et al.~\cite{KKST17} proved that every class of $(K_{3}\boxminus S_{3}, K_{3}\boxminus K_{3})$-free graphs of unbounded mim-width has unbounded sim-width. Recall that the class of chordal bipartite graphs has unbounded mim-width~\cite{BCM15}. It remains to observe that every chordal bipartite graph is $(K_{3}\boxminus S_{3}, K_{3}\boxminus K_{3})$-free. \qed
\end{proof}

\noindent
Theorem~\ref{t-chordalbip} strengthens the aforementioned result that chordal bipartite graphs have unbounded mim-width~\cite{BCM15}.

\subsection{The Missing Relationship in Figure~\ref{f-power}}\label{s-new}

A \emph{path decomposition} of a graph $G = (V,E)$ is a sequence
of subsets of vertices whose union is $V$ and such that: (1) for each edge $vw \in E$, there exists a subset containing both $v$ and $w$;
and (2) for each $v \in V$ the subsets containing $v$ are consecutive in the sequence. The width of a path decomposition is one less than the maximum size of a subset.
The \emph{path-width} of a graph $G$, denoted $\pw(G)$,
is the minimum possible width over all possible path
decompositions of~$G$.

Mannino, Oriolo, Ricci and Chandran~\cite{M-O-R-C-thinness} proved that $\thin(G)\leq \pw(G)+1$. We slightly modify their proof to show that also proper thinness is more powerful than path-width.

\begin{theorem} For a graph $G$,
$\pthin(G) \leq 2^{\pw(G)}(\pw(G)+1)$. Moreover, given a path
decomposition of width $q$, a vertex ordering and a strongly
consistent partition into at most $2^q(q+1)$ independent sets can be found
in polynomial time.
\end{theorem}

\begin{proof} In~\cite{M-O-R-C-thinness}, it is proved that, given a
path decomposition of width $q$, one can find a vertex ordering and a strongly
consistent partition into at most $q+1$ independent sets (a
colouring) in polynomial time, with the additional
property that each vertex has at most one neighbour smaller than
itself of each colour. By consistency, that possible neighbour is
the greatest vertex smaller than itself of that colour, if such a
vertex exists.

We refine that partition to make it strongly consistent with the
order, splitting each colour class into at most $2^q$ sets
according to whether or not it has a neighbour smaller than itself in each
of the other colour classes. Notice that refining a partition
maintains consistency, so, in order to prove strong consistency,
let $u < z < v$, $uv \in E(G)$, with $z, v$ in the same refined set.
Since the colour classes are independent sets, $u$ and $v$ are of
distinct colours, say $a$ and $b$, respectively, and $u$ is the
greatest vertex smaller than $v$ of colour $a$. By the way of
refining the colour classes, $z$ is of colour $b$ and since $v$
does have a neighbour of colour $a$ smaller than itself, so does
$z$. Since $u$ is the greatest vertex smaller than $z$ of
colour $a$, $uz \in E(G)$. So, we have defined a vertex ordering
and a partition into at most $2^q(q+1)$ independent sets that are
strongly consistent and can be found in polynomial time. \qed
\end{proof}

\section{The Proof of Theorem~\ref{t-lc}}\label{s-new2}

In this section we prove Theorem~\ref{t-lc}.
Let $G$ be a graph. A function $c:V(G)\to \{1,2,\ldots\}$ is a {\it colouring} of $G$ if $c(u)\neq c(v)$ for every pair of adjacent vertices $u$ and $v$.
A {\it list assignment} of a graph $G=(V,E)$ is a function~$L$ that prescribes a list of ``admissible'' colours $L(u)\subseteq \{1,2,\ldots\}$ to each $u\in V$.
A colouring $c$  {\it respects}~${L}$ if  $c(u)\in L(u)$ for every $u\in V$.
For an integer~$k\geq 1$, if $L(u)\subseteq \{1,\ldots,k\}$ for each $u\in V$, then $L$ is a {\it list $k$-assignment}.
The {\sc List $k$-Colouring} problem is to decide if a graph $G$ with a list $k$-assignment $L$ has a colouring that respects $L$.  Note that if every list is $\{1,\ldots,k\}$, we obtain the classical {\sc $k$-Colouring} problem.

We need the following well-known result regarding the {\sc $2$-List Colouring} problem. For this problem, a graph~$G$ with a list assignment $L$ is given as input, where $|L(u)|\leq 2$ for every $u\in V(G)$.

\begin{theorem}[\cite{Ed86}]\label{t-2sat}
The {\sc $2$-List Colouring} problem is linear-time solvable.
\end{theorem}

\noindent
We are now ready to prove the main result of this section.

\medskip
\noindent
{\bf Theorem~\ref{t-lc} (restated).} {\it For $k\geq 4$, {\sc List $k$-Colouring} is \NP-complete for star convex graphs and comb convex graphs, while {\sc List $3$-Colouring} is polynomial-time solvable for star convex graphs.}

\begin{proof}
We start with proving the hardness results.
It is well-known that for $k\geq 3$, {\sc List $k$-Colouring} is \NP-complete for bipartite graphs, even for various subclasses of bipartite graphs (see, for example,~\cite{CC06,HT96}).
Hence, let $k\geq 3$ and $(G,L)$ be an instance of {\sc List $k$-Colouring}, where
$G$ is a bipartite graph with partition classes $A$ and $B$.

We let $G^*$ be the bipartite graph obtained from $G$ by adding $|A|$ new vertices to $A$ that we make complete to $B$.
Then, by Lemma~\ref{constr}, $G^*$ is comb convex. Let $L^*$ be the extension of $L$ where each new vertex has been given the list $\{k+1\}$. We observe that $(G,L)$ is a yes-instance of {\sc List $k$-Colouring} if and only if
$(G^*,L^*)$ is a yes-instance of {\sc List $(k+1)$-Colouring}.

We let $G'$ be the bipartite graph obtained from $G$ by adding one new vertex to $A$ that we make complete to $B$.
By Lemma~\ref{constr}, $G'$ is star convex. We give the new vertex the list $\{k+1\}$ and repeat the argument above.

We now prove the polynomial-time result, which only holds for star convex graphs. So consider a star convex graph $G=(A,B,E)$ with a list $3$-assignment~$L$.
We may assume without loss of generality that for every $u\in B$, it holds that $L(u)$ has size at least~$2$; else if $L(u)=\{i\}$ for some $u\in B$ and
$i\in \{1,2,3\}$, then we delete colour~$i$ from the list of every neighbour of $u$ in $A$ and delete $u$ from $G$.

Let $a\in A$ correspond to the center of the star that is a support for $G$. Then, every vertex in $B$ is either adjacent to $a$ or has degree~$1$. We may remove every vertex $u\in B$
of degree~$1$ from $G$: as $|L(u)|\geq 2$, we are always able to assign $u$ a colour from $L(u)$ after restoring it. Afterward removing vertices of $B$ that have degree~$1$ from $G$, every vertex of $B$ is adjacent to~$a$.

We now consider all of the at most three options of colouring $a$ with a colour from $L(a)$. For each chosen colour $c(a)\in L(a)$, we remove $c(a)$ from the list of every vertex of $B$ and moreover, we assign $c(a)$ to every other vertex $a'\in A$ with $c(a)\in L(a')$. Afterwards we obtain an instance of {\sc $2$-List Colouring}, which means we can apply Theorem~\ref{t-2sat}.
\qed
\end{proof}

\section{Final Remarks}\label{s-final}

In this paper we generalized and unified a number of algorithmic results for generalized convex graphs by showing boundedness of mim-width.
We are not aware of any new algorithmic implications due to our refined width parameter analysis in Section~\ref{s-refined}.
We do observe that there exist problems that are \NP-complete for graphs of bounded (linear) mim-width, but polynomial-time solvable for graphs of bounded thinness. Namely, Vatshelle~\cite{Va12} proved that {\sc Clique} (the problem of deciding if a graph has a clique of size at least~$k$ for some given integer~$k$) is \NP-complete for graphs of mim-width at most~$6$; in fact the branch decomposition in Vatshelle's proof turns out to be even linear. On the other hand, {\sc Clique} belongs to a large framework of graph problems that are polynomial-time solvable for graphs of bounded thinness~\cite{BE19}. However, {\sc Clique} is trivial on bipartite graphs, and all the graph classes we consider in this paper are bipartite. It would therefore be interesting to research if there are natural problems that are \NP-complete for bipartite graphs of bounded mim-width but  polynomial-time solvable for graphs of bounded thinness or bounded linear mim-width. We are also not aware of any problems that are \NP-complete for graphs of bounded thinness, but polynomial-time solvable for graphs of bounded proper thinness.

For answering the above questions, more structural results may be needed. For instance, finding classifications of $(H_1,H_2)$-free graphs and $H$-free (circular) convex graphs of bounded thinness might help. Moreover, can we characterize convex graphs of proper thinness at most~$k$ and  circular convex graphs of thinness at most~$k$ by some obstruction set, for example, a forbidden set of induced subgraphs?

In addition, it would also be interesting to obtain dichotomies for more graph problems restricted to classes of generalized convex graphs. We ask this question in particular for graph problems known to be solvable in polynomial time for graph classes whose mim-width is bounded and quickly computable.
In our paper we gave examples of ten of such problems (see Corollaries~\ref{c-2} and~\ref{c-22}).

Generalized convex graphs also play a role in other settings.
For example, Chen et al.~\cite{CK15} considered the problem {\sc Subset Interconnection Design}, which is
to decide if a bipartite graph belongs to a class of ${\cal H}$-convex graphs. This problem and its variants have several applications,
for example in the design of scalable overlay networks and vacuum systems~\cite{CK15}, combinatorial auctions~\cite{GG07} and fair allocation of indivisible goods~\cite{BCE17}. Are the problems in these settings solvable for graph classes whose mim-width is bounded and quickly computable? We leave this for future research as well.

In a recent arXiv paper, Jaffke, Kwon and Telle~\cite{JKT21} introduced the notion of bi-mim-width for directed graphs. As a consequence of their study, they considered {\it $H$-convex graphs} for a fixed graph~$H$, which in our terminology corresponds to ${\cal H}$-convex graphs where ${\cal H}$ consists of all subdivisions of $H$. For example, when $H$ is the cycle on two vertices with two edges between them, we obtain the class of circular convex graphs.
They showed that the linear mim-width of an $H$-convex graph is at most $6|E(H)|$.
For circular convex graphs this leads to a bound of~$12$.
If $H$ is a tree of maximum degree at most~$\Delta$ with at most $t$ vertices of degree at least~$3$, then the bound of $6|E(H)|$ leads to a bound of
$6(t\Delta - t + 1)$, as $H$ has at most $t\Delta - t + 1$ edges.
Note that combining Theorem~\ref{t-thin} with Lemma~\ref{l-mwt} yields a bound of $2+t(\Delta-2)$.

We finish our paper with some open problems on {\sc List $k$-Colouring}. The first open problem results from Section~\ref{s-new2}: what is the computational complexity of {\sc List $3$-Colouring} for comb convex graphs? Our second open problem was also asked by Huang et al.~\cite{HJP15}, who proved that for all $k\geq 4$, {\sc List $k$-Colouring} is \NP-complete
for $P_8$-free chordal bipartite graphs. What is the computational complexity of {\sc List $3$-Colouring} for chordal bipartite graphs? Answering both questions would complete the results for {\sc List $k$-Colouring} for the graph classes displayed in Figure~\ref{fig:dia}.

\end{document}